\documentclass{hjm1}
\usepackage{verbatim}

\makeatletter
\newcommand{\eqmathbox}[2][N]{\eqmakebox[#1]{$\displaystyle#2$}}
\usepackage[all]{xy}
\usepackage{tikz,tikz-cd}
\usetikzlibrary{matrix, arrows}
\usepackage{subcaption}
\usepackage{adjustbox}
\usepackage{eqparbox}
\usepackage{graphicx}%
\usepackage{multirow}%
\usepackage{amsmath,amssymb,amsfonts}%
\usepackage{amsthm}%
\usepackage{mathrsfs}%
\usepackage[title]{appendix}%
\usepackage{xcolor}%
\usepackage{textcomp}%
\usepackage{manyfoot}%
\usepackage{booktabs}%
\usepackage{algorithm}%
\usepackage{algorithmicx}%
\usepackage{algpseudocode}%
\usepackage{listings}%

\begin{document}

\title[Diffraction and self-similar fractal strings]
{Diffraction measures and patterns of the complex dimensions of self-similar fractal strings. I. The lattice case}
\author[M. L. Lapidus, M. van Frankenhuijsen, and E. K. Voskanian]
{Michel L. Lapidus, Machiel van Frankenhuijsen, \\ and Edward K. Voskanian}
\address [Michel L. Lapidus]{Department of Mathematics, University of California, Riverside, 900 University
Ave, Riverside, 92521, California, USA.} \email{lapidus@math.ucr.edu}\
\address[Machiel van Frankenhuijsen]{Department of Mathematics, Utah Valley University, 800 W University
Pkwy, Orem, 84058, Utah, USA.} \email{vanframa@uvu.edu}
\address[Edward K. Voskanian]{Department of Mathematics, Norwich University, 158 Harmon Dr,
Northfield, 05663, Vermont, USA.} \email{evoskani@norwich.edu}

\keywords{Quasicrystals, fractal geometry, complex dimensions, lattice and nonlattice self-similar fractal
strings, mathematical diffraction, discrete diffraction pattern, autocorrelation measure,
diffraction measure.}

\subjclass[2000]{52C23; 28A80; 28A33}

\begin{abstract}
We give a generalization of Lagarias' formula for diffraction by ideal crystals, and we apply it to the lattice case, in preparation for addressing the problem of quasicrystals and complex dimensions posed by Lapidus and van Frankenhuijsen concerning the quasiperiodic properties of the set of complex dimensions of any nonlattice self-similar fractal string. More specifically, in this paper, we consider the case of the complex dimensions of a lattice (rather than of a nonlattice) self-similar string and show that the corresponding diffraction measure exists, is unique, and is given by a suitable {\it continuous} analogue of a discrete Dirac comb. We also obtain more general results concerning the autocorrelation measures and diffraction measures of generalized idealized fractals associated to possibly degenerate lattices and the corresponding extension of the Poisson Summation Formula.
\end{abstract}
 
\maketitle
 
\theoremstyle{plain}%
\newtheorem{theorem}{Theorem}
\newtheorem{proposition}[theorem]{Proposition}%

\theoremstyle{remark}%
\newtheorem{example}{Example}%
\newtheorem{remark}{Remark}%
\newtheorem{problem}{Problem}%
\newtheorem{corollary}[theorem]{Corollary}%

\theoremstyle{definition}%
\newtheorem{definition}{Definition}%

\section{Introduction} \label{introduction}

Far field diffraction (or kinematic diffraction) is when light waves from a far away source are pointed at, say, an alloy material, which are then scattered (or diffracted) by the alloy's constituent atoms. The diffracted light recombines and produces an interference (or diffraction) pattern; see, e.g., \cite{lipson2010optical,cowley1995diffraction}). A quasicrystal has a diffraction pattern consisting of (essentially) bright spots, and with rotational symmetries that are incompatible with periodic atomic structures (see, e.g., \cite[Theorem 2.1]{senechal1996quasicrystals}). Put another way, a quasicrystal is a nonperiodic material with enough order for giving rise to discrete diffraction patterns. Up until the discovery of quasicrystals by Dan Shechtman \cite{blech1984metallic} in 1982,\footnote{The first known quasicrystal was an alloy processed in a lab. Over twenty years after Shechtman's discovery, the first naturally occurring quasicrystal was reported to be found in the Koryak Mountains in Russia \cite{bindi2009natural,bindi2011icosahedrite}.} such materials were believed to be impossible matter. The emergence of quasicrystals led to new mathematical ideas to understand their formation (see, e.g., \cite{baake2013aperiodic,baake2017aperiodic,lapidus2008in,senechal1996quasicrystals}, and the references therein). The present paper is concerned with a particular mathematical analogue of kinematic diffraction (explained in Section \ref{background}) that we use to study potential mathematical models for quasicrystals associated to fractals.

A set $X \subset \mathbb{R}^n$ is called {\it uniformly discrete} if there exists a constant $r > 0$ such that for every $x \in X$, $\{y \in X\colon \left\lvert y - x \right\rvert < r\} = \{x\}$. A set $X \subset \mathbb{R}^n$ is called {\it relatively dense} if there exists a constant $R > 0$ such that for all $x \in \mathbb{R}^n$, $\{y \in X\colon \left\lvert y - x \right\rvert \leq R\} \neq \emptyset$. Any set $X \subset \mathbb{R}^n$ that is both uniformly discrete and relatively dense is called a {\it Delone} set. Cut-and-project sets (or model sets) are the special kinds of Delone sets that are formed by taking a slice of a lattice (see Definition \ref{lattice def}) and projecting it down into lower dimensions. They are the most popular mathematical models for quasicrystals (see, e.g. Moody's survey on model sets, which appeared in \cite{axel2000from}). It is well-known that the diffraction measure for a cut-and-project set is discrete. More recent results on mathematical diffraction by cut-and-project sets can be found in, for example, see \cite{lenz2007pure, richard2017pure,richard2017a} and the relevant references therein. 
  
In the aforementioned mathematical approach to diffraction, a quasicrystal is defined by the fact that the diffraction measure is purely discrete (i.e., is a Dirac comb); see, e.g., Senechal's book \cite[Chapter 3.5]{senechal1996quasicrystals}. In \cite[Problem 3.22, p. 89]{lapidus2012fractal}, partly motivated by \cite{lapidus2008in}, the first two authors have asked whether there exists an appropriate notion of a generalized quasicrystal according to which the complex dimensions of a nonlattice self-similar string are such a (generalized) quasicrystal.  

In the present paper, we show, in particular, that the complex dimensions of {\it lattice} self-similar strings form a generalized quasicrystal, in the sense that their diffraction measure is well defined and is a kind of (generalized) {\it continuous} (rather than discrete) {\it Dirac comb}. We establish this statement by first extending to the case of degenerate lattices Lagarias' formula for the diffraction measure of an idealized quasicrystal \cite[Theorem 2.7]{lagarias2000mathematical}. 

The structures considered in the present paper are not Delone sets, as they lack the property of relative density. Our goal is to compute diffraction measures for sets that do not fill up all of Euclidean space, which is motivated by the need to understand the diffraction properties of the sets of complex dimensions of self-similar fractal strings. In the literature, there is some work focused on computing diffraction patterns of fractal sets (see, e.g., \cite{godreche1990multifractal,dettmann1993structure,baake2020fourier}, and the relevant references therein). Our work is not focused on diffraction patterns of fractal sets. Instead, we are interested in diffraction patterns from complex dimensions of certain fractals, which are discrete subsets of two-dimensional Euclidean space.  

Since, according to \cite[Theorem 3.18, p. 84]{lapidus2012fractal} (see also \cite{lapidus2003complex}), nonlattice self-similar strings and their complex dimensions can be approximated arbitrarily closely by a sequence of lattice strings which, therefore, exhibit quasiperiodic patterns (see \cite[Chapter 3, esp., Section 3.4]{lapidus2012fractal}) --- and similarly, for the associated complex dimensions --- it is expected that our present results should help us in future work (see, e.g., \cite{lapidus2024diffraction}) to address the aforementioned open problem about the generalized ``quasicrystality'' of the quasiperiodic patterns of complex dimensions of nonlattice self-similar strings.

\subsection{Ideal crystals} 
\hfill

\noindent The following well-known definition can be found, e.g., in Bremner's book \cite[Definition 1.15]{bremner2011lattice}.  
\begin{definition} \label{lattice def}
Let $x_1, x_2, \dots, x_m$ be a set of $m \geq 1$ linearly independent vectors in $\mathbb{R}^m$. The {\it $m$-dimensional lattice} spanned by these vectors is
\[ \Lambda = \left\{\sum_{j = 1}^m a_jx_j\colon a_1, a_2, \dots, a_m \in \mathbb{Z} \right\}. \]
\end{definition}
For $j = 1, \dots, m$ we write $x_{j} = \left(x_{j1}, x_{j2}, \dots, x_{jm}\right)$ and form the $m \times m$ {\bf basis matrix} $B = \left(x_{jk} \right)$. The {\it determinant} of the lattice $\Lambda$, denoted by $\det(\Lambda)$, is given by
\[ \det(\Lambda) = \lvert \det(B) \rvert. \]
Note that the determinant does not depend on the choice of basis \cite[Exercise 1.12]{bremner2011lattice}. Geometrically, the determinant is the volume of the parallelepiped formed by the vectors $x_1, x_2, \dots, x_m$.

In \cite[Definition 1.3]{lagarias2000mathematical}, Lagarias defines an {\it ideal crystal} in $\mathbb{R}^m$ to be a finite number of translates of an $m$-dimensional lattice in $\mathbb{R}^m$. That is, if $\Lambda \subset \mathbb{R}^m$ is a lattice and $F$ is a finite subset of $\mathbb{R}^m$, the set
\[ X = \Lambda + F = \{y + f\colon y \in \Lambda, f \in F\} \]
is called an {\it ideal crystal}. An ideal crystal is a kind of generalized lattice, in that it is an infinite repetition of a single point, or of a finite collection of distinct points; see Figure \ref{ideal crystal full rank}.
\begin{figure}[tb]
\centering
\begin{subfigure}{0.9\textwidth}
    \centering
    \includegraphics[width=0.9\textwidth]{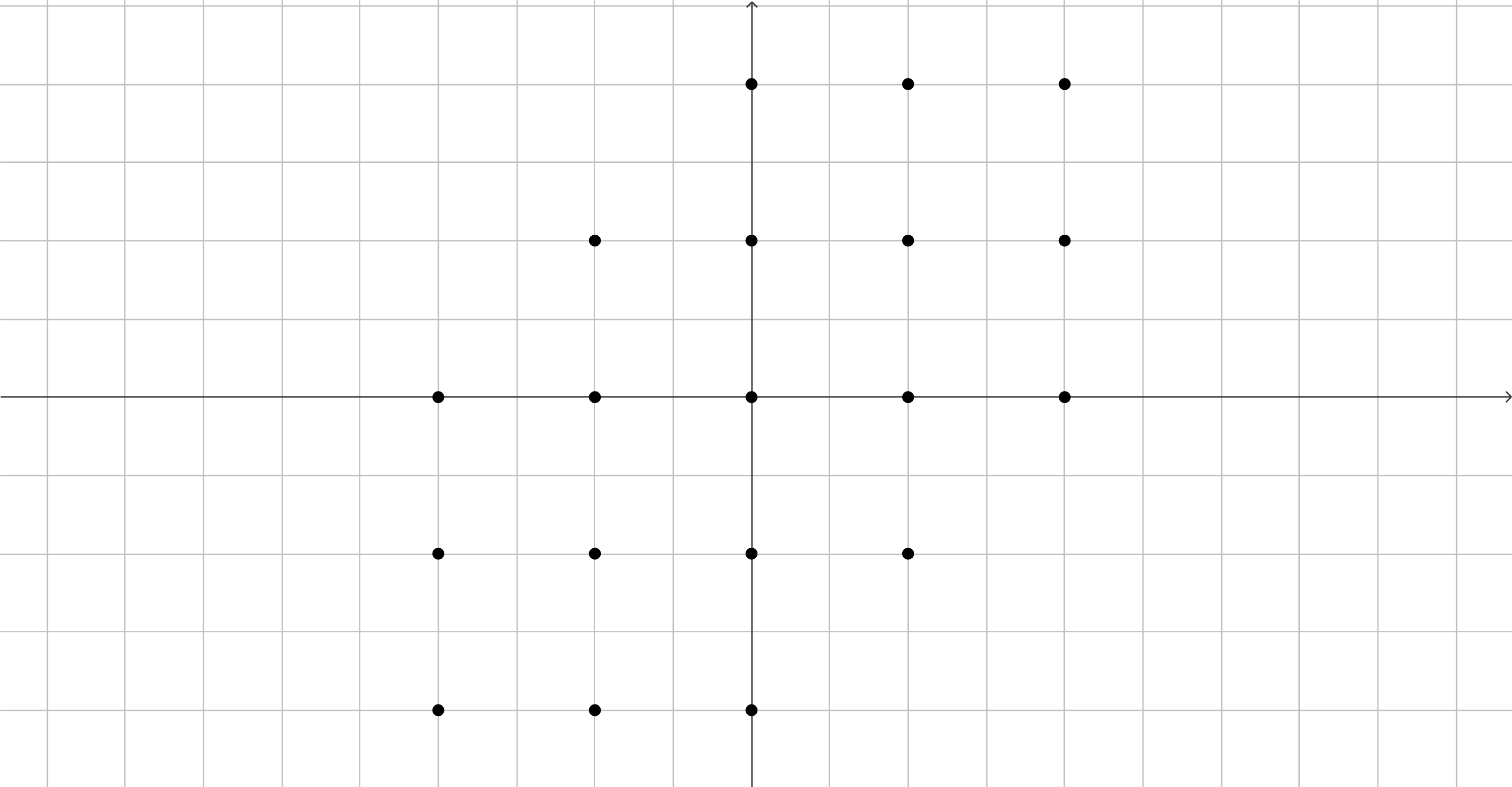}
    \caption{}
\end{subfigure}
\hfill
\begin{subfigure}{0.9\textwidth}
    \centering
    \includegraphics[width=0.9\textwidth]{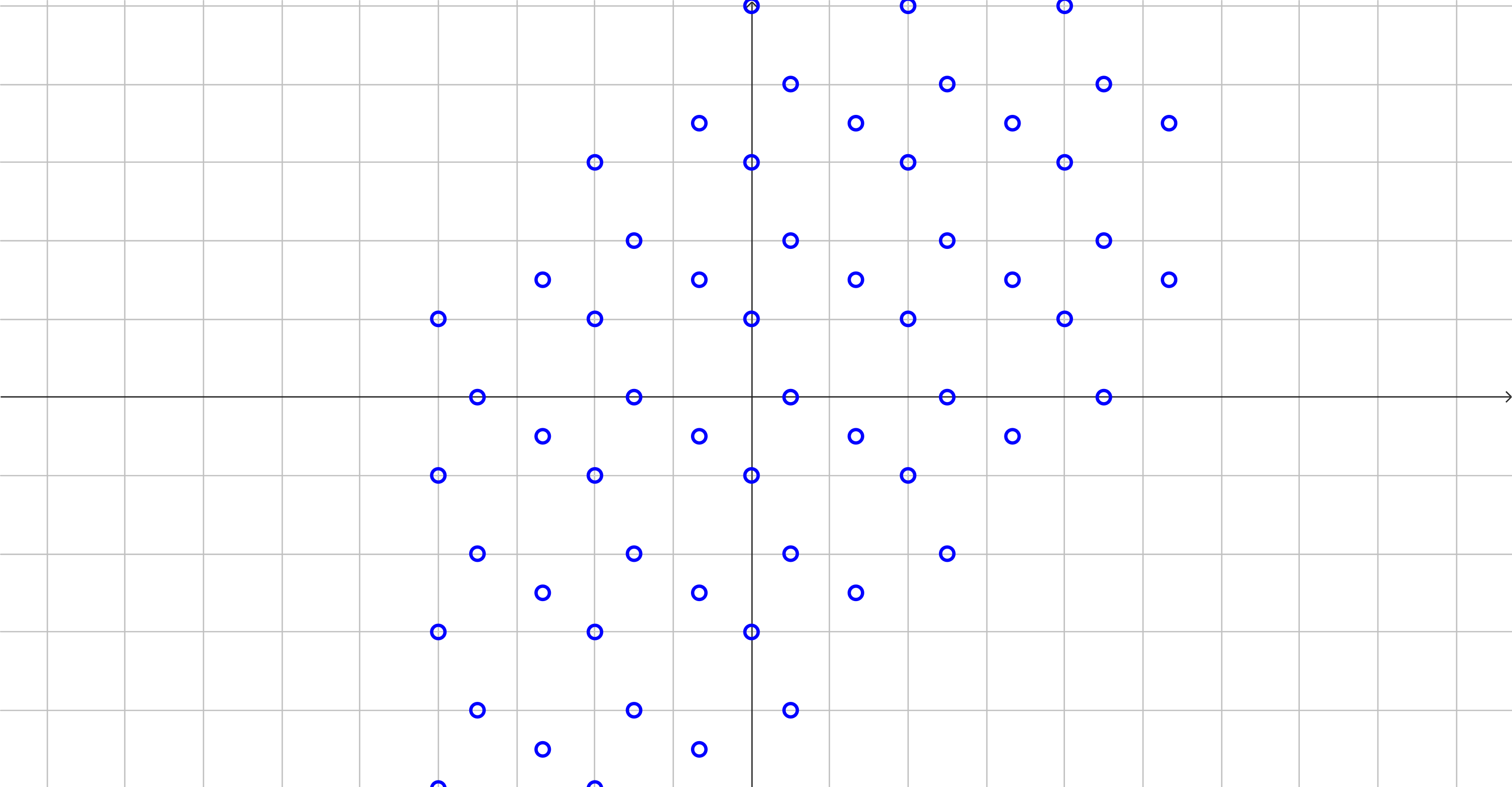}
    \caption{}
    \label{fig:full rank}
\end{subfigure}
\caption{(a) Part of a lattice in the plane; and (b) the ideal crystal from this lattice (from part (a)), along with the finite set $F =~ \{(4/3,3/2),(1/2,0),(0,1)\}$.}
\label{ideal crystal full rank}
\end{figure}
Lagarias views ideal crystals as idealized infinite periodic materials, and from this perspective, he has shown, using a certain mathematical analogue of diffraction, that ideal crystals have discrete mathematical diffraction patterns \cite[Theorem 2.7]{lagarias2000mathematical}; that is, their diffraction measure exists and is equal to a suitable Dirac comb (i.e., a finite or countable weighted sum of Dirac measures). 

\subsection{Generalized Ideal Crystals} \label{main results}
In the present paper, we study translations in Euclidean space of an $m$-dimensional lattice in $d + m$ dimensions (see Definition \ref{genIdeal} and Figure \ref{ideal crystal}). 
\begin{figure}[tb]
\centering
\begin{subfigure}{0.45\textwidth}
    \centering	
    \includegraphics[width=0.95\textwidth]{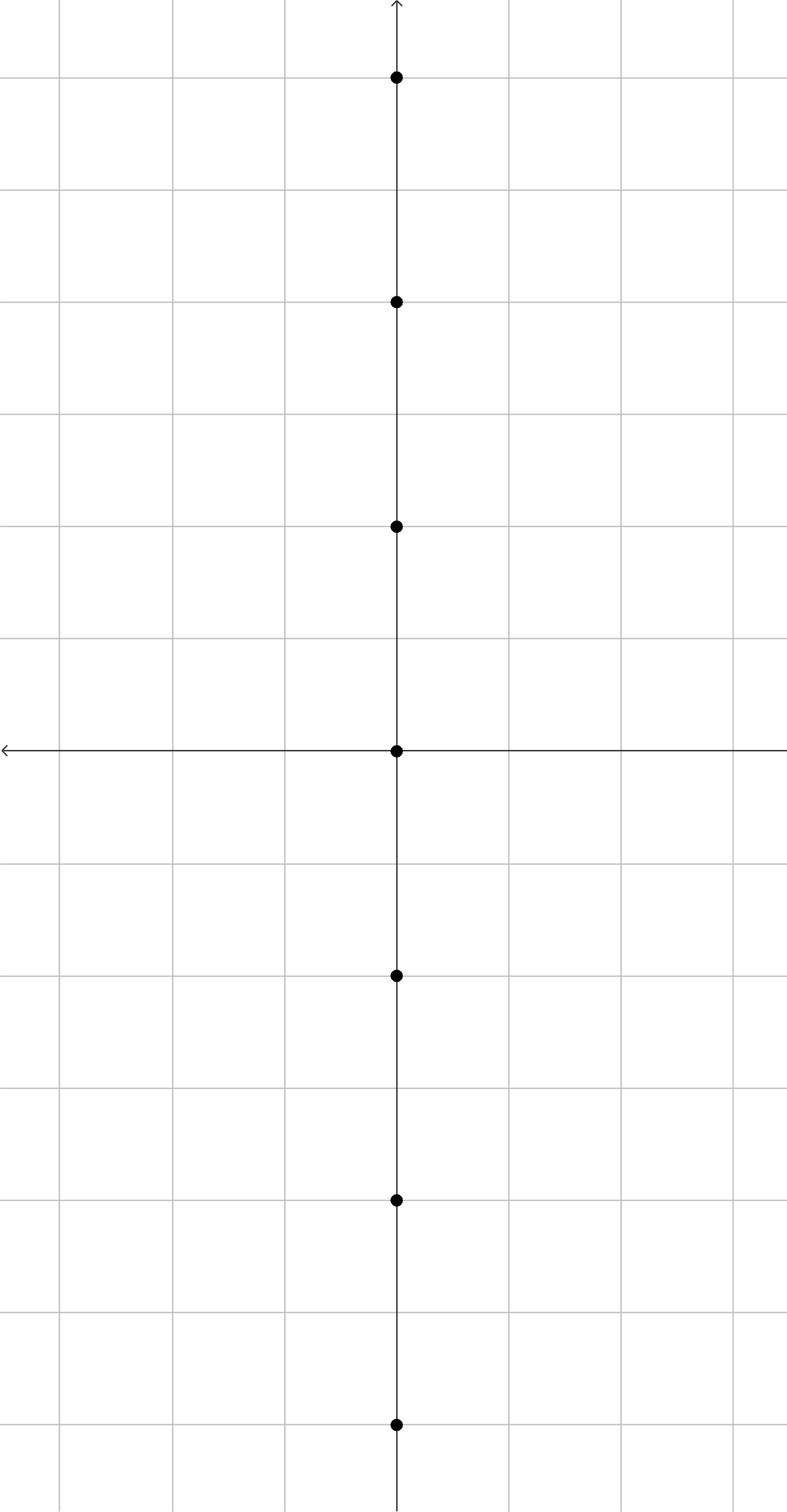}
    \caption{}
\end{subfigure}
\hfill
\begin{subfigure}{0.45\textwidth}
    \centering
    \includegraphics[width=0.95\textwidth]{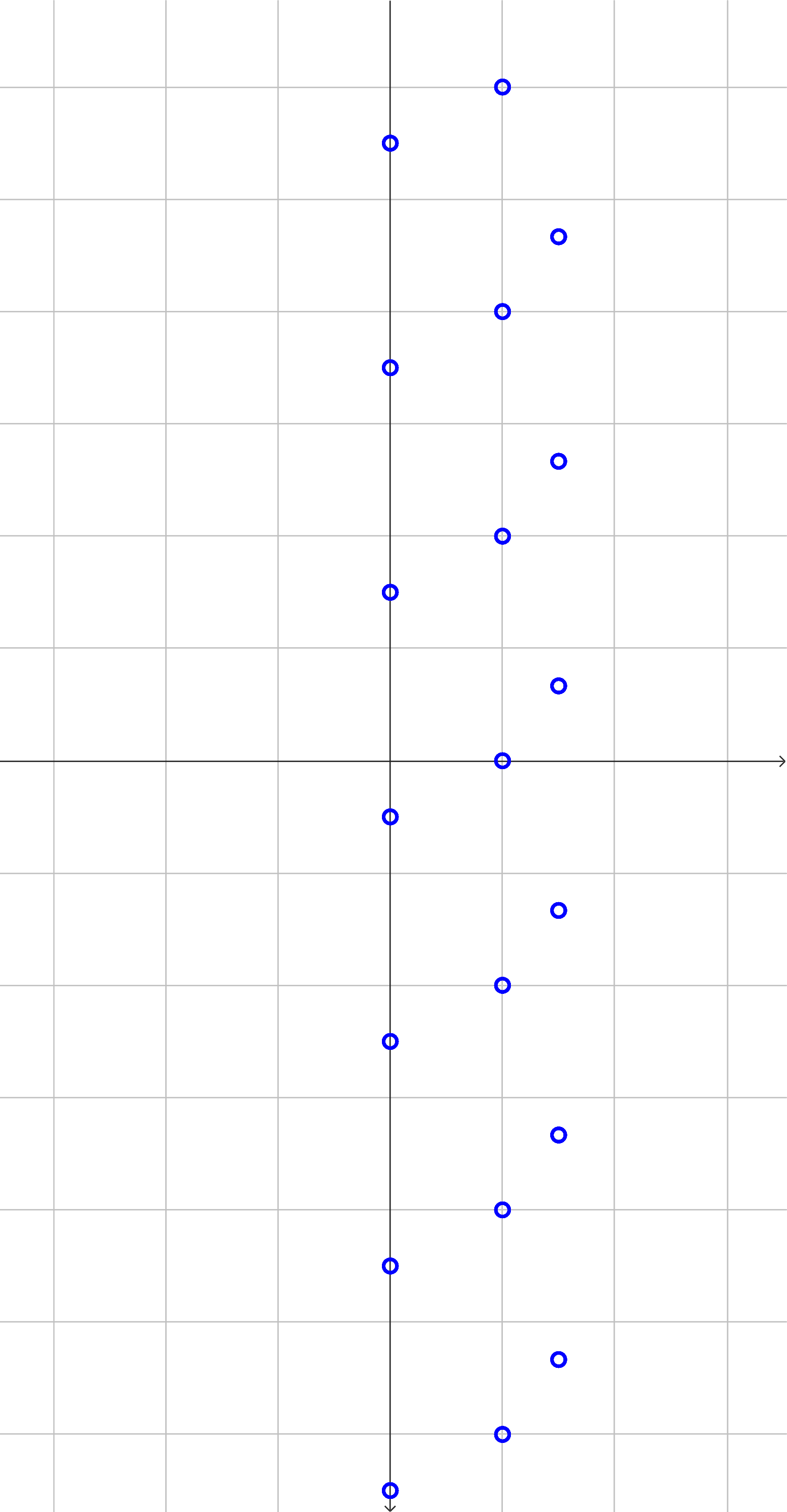}
    \caption{}
    \label{fig:rank-1}
\end{subfigure}
\caption{(a) Part of a $1$-dimensional lattice in the plane; and (b) the ideal crystal from this lattice (from part (a)), along with the finite set $F =~ \{(3/2,-5/4),(0,-1/2),(1,0)\}$.}
\label{ideal crystal}
\end{figure}
We do this by first taking a lattice $\Lambda \subset \mathbb{R}^m$ and identifying it with its image $T(\Lambda)$ under the map $T: \mathbb{R}^{m} \to \mathbb{R}^{d + m}$ given by 
\begin{equation} \label{id} 
T(a) = (\underbrace{0, \dots, 0}_{d\texttt{-}\text{times}},a_1, a_2, \dots, a_m),\text{ for all }(a_1, \dots, a_m) \in \mathbb{R}^m.
\end{equation}
\begin{definition} \label{genIdeal}
Let $\Lambda$ be a lattice in $\mathbb{R}^m$. Let $d$ be a nonnegative integer, and let $F \subset \mathbb{R}^{d + m}$ be a finite set of translation vectors. We call the set
\[ X = T(\Lambda) + F \]
a $(d + m)$-dimensional {\it ideal crystal} (or a possibly degenerate {\it ideal quasicrystal\/}), and we say that the ideal crystal has {\it rank} equal to $m$.
\end{definition}
By utilizing a suitable noncompact version of the classical Poisson Summation Formula (see Theorem \ref{LatticegeneralPSF}), we adapt Lagarias' theorem to the more general ideal crystals defined above (see Theorem \ref{idealDiff}).   

Our work is motivated by an open problem \cite[Problem 3.22]{lapidus2012fractal} (see Problem~\ref{open problem} below), stated by the first two authors, which connects the mathematical theory of quasicrystals and aperiodic order to the one-dimensional theory of fractal geometry and complex dimensions. Lagarias states \cite[p. xi]{baake2017aperiodic}: ``The mathematical study of aperiodically ordered structures is a beautiful synthesis of geometry, analysis, algebra and number theory.'' The sequel \cite{lapidus2024diffraction} of the present paper will show a deep connection between Problem \ref{open problem} and the long-standing problem of simultaneous Diophantine approximation (see, e.g., \cite{bremner2011lattice,lagarias1985the,lagarias1982best,lapidus2003complex,lapidus2012fractal,lenstra1982factoring,lapidus2021quasiperiodic}).  

\section{Diffraction and Complex Dimensions of Fractal Strings} \label{background}

In this section, we describe a recipe for mathematical diffraction from \cite[Example 2.1]{hof1995on}, and in Section \ref{results}, and we adapt it to ideal crystals in $\mathbb{R}^{d + m}$ with rank $m$ ($d > 0$). We refer to these as lower-rank ideal crystals. Note that the theory of mathematical diffraction used in this paper was initiated by Dworkin \cite{dworkin1993spectral}, and then extended by Hof \cite{hof1992quasicrystals,hof1995on,hof1995diffraction,hof1997diffraction}. For more information on the mathematical theory of diffraction, see, e.g., \cite{senechal1996quasicrystals}, \cite[Chapter 9]{baake2013aperiodic} and \cite{baake2017aperiodic}, along with the references therein. We also give a brief summary of the parts of the one-dimensional theory of fractal strings and complex dimensions which are directly related to the goal of the present paper and its sequel, which was mentioned at the end of Section \ref{introduction}. 

\subsection{Mathematical diffraction}
\hfill

\noindent The Fourier transform plays a central role in physical diffraction. If $\rho$ is a function describing the atomic nature of a material, then its diffraction pattern is described by $\left\lvert \widehat{\rho} \right\rvert^2$, where $\widehat{\rho}$ denotes the Fourier transform of $\rho$. This is summarized by the following {\it Wiener diagram}:
\begin{equation}
\begin{tikzcd}[row sep=huge,column sep=huge]
\eqmathbox{\rho} \arrow{dr}{} \arrow[d,"\widehat{\ }"] \arrow[r,"\ast"] & \eqmathbox[M]{\rho \ast \widetilde\rho} \arrow[d,"\widehat{\ }"]\\
\eqmathbox{\widehat{\rho}} \arrow[r, "\left\lvert \cdot \right\rvert^2"] & \eqmathbox[M]{\left\lvert \widehat{\rho} \right\rvert^2}
\end{tikzcd} 
\label{wiener}
\end{equation}
where $\rho$ is integrable, $\widetilde{\rho}(x) = \overline{\rho(-x)}$, $\widehat{\rho}$ is the Fourier transform of $\rho$, as before, and $\rho \ast \widetilde{\rho}$ is the convolution product of $\rho$ and $\widetilde{\rho}$. Note that this diagram commutes as a result of the convolution theorem and the identity $\widehat{\widetilde{\rho}} = \overline{\widehat{\rho}}$, the complex conjugate of $\widehat{\rho}$.   

Let $X$ be a (countably infinite) set in $\mathbb{R}^n$ with the property that for any compact set $K \subset \mathbb{R}^n$, there exists a positive and finite constant $a_K$ such that the number of points in 
\[ x + K := \{x + k\colon k \in K\} \] 
is bounded from above by $a_K$, and that we have the same uniform bound, regardless of the choice of $x \in X$. Mathematical diffraction follows the path from $\rho$ to $\left\lvert \widehat{\rho} \right\rvert^2$ in the Wiener diagram, except for the fact that $\rho$ is now replaced with a measure, defined below. In \cite{hof1995on}, Hof justifies why the path through the autocorrelation $\rho \ast \widetilde{\rho}$ should be taken.

Denote by $C_c(\mathbb{R}^n)$ the vector space of complex-valued continuous functions on~$\mathbb{R}^n$ with compact support. 

\begin{definition}
A ({\it complex}) {\it measure} $\mu$ on $\mathbb{R}^n$ is a linear functional on $C_c(\mathbb{R}^n)$ with the property that for every compact subset $K \subset \mathbb{R}^n$, there exists a positive and finite constant $b_K$ such that
\[ \left\vert \mu(f) \right\vert \leq b_K\sup\{\left\lvert f(x) \right\rvert\colon x \in K \}, \]
for all $f \in C_c(\mathbb{R}^n)$ with support contained in $K$. As a result, of course, $\mu$ is nothing else but a {\it continuous} linear functional on $C_c(X)$.

The {\it conjugate} of the measure $\mu$, denoted by $\overline{\mu}$, is another measure given by
\[ \overline{\mu}(f) =  \overline{\mu(\overline{f})},\text{ for all }f \in C_c(\mathbb{R}^n). \]
The measure $\mu$ is said to be {\it real} if $\overline{\mu} = \mu$, and a real measure is called {\it positive} if $\mu(f) \geq 0$, for all nonnegative functions $f \in C_c(\mathbb{R}^n)$. The latter condition guarantees the continuity of $\mu$. Every measure $\mu$ has an associated positive measure, denoted by $\left\lvert \mu \right\rvert$, called the {\it total variation measure} of $\mu$, with the property that for every $f \in C_c(\mathbb{R}^n)$ such that $f(x) \geq 0$ for all $x \in \mathbb{R}^n$, we have that
\[ \left\lvert \mu(f) \right\rvert \leq \left\lvert \mu \right\rvert(f). \]
\end{definition}

In the context of mathematical diffraction, the function $\rho$ in the Wiener diagram (\ref{wiener}) is replaced with the measure
\[ \mu := \sum_{x \in X} \delta_x, \]
which serves as an analytical representation of the diffracting material $X$. The measure $\mu$ is commonly referred to as a {\it Dirac comb}, and it plays the role of the function $\rho$ in the Wiener diagram. Note that in the present paper, we only consider diffraction by sets $X$ that are sufficiently regular so that $\mu$ is always a tempered distribution. In the following subsection, we describe the analogue (in the present context of mathematical diffraction) of the autocorrelation in the Wiener diagram, which is $\rho \ast \tilde{\rho}$ in the upper right-hand corner of (\ref{wiener}). 

\subsubsection{The Autocorrelation Measure} \label{auto section}

The autocorrelation measure in the Wiener diagram~(\ref{wiener}) describes the nature of the displacement vectors between atoms. It is worth noting that in physics, the inter-atomic interactions between atoms is largely influential in the types of diffraction patterns obtained. For any continuous function $f$, the function $\widetilde{f}$ is given by $\widetilde{f}(x) := \overline{f(-x)}$, for any $x \in \mathbb{R}^n$. The natural extension to measures is given by $\widetilde{\mu}(f) := \mu(\tilde{f})$, for all $f \in C_c(\mathbb{R}^n)$. The convolution, $\mu \ast \nu$, of two (suitable) measures $\mu$ and $\nu$ is given by the integral
\[ (\mu \ast \nu)(f) = \int f(x + y)\, d\mu(x)\, d\nu(y), \]
whenever it exists. We note that a finite measure is necessarily translation bounded (defined just below), and that the convolution of two measures is again a measure if at least one of the measures is finite with the other being only translation bounded.

\begin{definition}
A measure $\mu$ is said to be {\it translation bounded} if for every compact $K \subset \mathbb{R}^n$, 
\[ \sup_{x \in X} \left\lvert \mu \right\rvert(x + K) < \infty. \] 
\end{definition}

Let $C_L$ denote the centered $n$-cube with side length equal to $L$ (and center the origin in $\mathbb{R}^n$). For example, in $\mathbb{R}$, $C_2 = [-1,1]$. We denote the restriction of a measure $\mu$ to $C_L$ by $\mu_L$. Since $\mu_L$ is a finite measure, the measure
\[ \gamma_{\mu}^L = \frac{\mu_L \ast \widetilde{\mu}_L}{L^n} \]
is well defined. Convergence of measures in the {\it vague sense} means that 
\[ \mu_j(f) \to \mu(f),\text{ as }j \to \infty, \] 
for all $f \in C_c(\mathbb{R}^n)$, and where $\mu$ and $\mu_j$ (for each $j \in \mathbb{N}$) are measures on $\mathbb{R}^n$. Furthermore, every vague limit point of the family $\{\gamma_{\mu}^L\colon L > 0\}$ is called an {\it autocorrelation measure} of $\mu$. In the present paper, all of the Dirac combs that are considered will have exactly one autocorrelation measure. 

Let $A := X - X$ denote the set $\{x - y\colon x, y \in X\}$ consisting of all possible displacement vectors between points of $X$, and assume that $A$ is {\it locally finite}. That is, assume that any compact subset $K \subset \mathbb{R}^n$ contains only finitely many points of $A$. For a fixed length $L$, we count the number of pairs $(x_1,x_2) \in C_L$ with the property that the displacement vector $x_2 - x_1$ from $x_1$ to $x_2$, is equal to a prescribed displacement vector $a \in A$. More precisely, we seek a formula for
\begin{equation} \label{count}
N_L(a) := \#\{(x_1,x_2) \in X \times X\colon x_1, x_2 \in C_L \text{ and } a = x_2 - x_1\},
\end{equation}
which is the number of ordered pairs in $C_{L}$ with displacement vector (or interaction) equal to $a$. For example, when we count pairs $(x_1,x_2)$ with displacement vector equal to zero, we are really just counting the number of points of $X$ in the $n$-cube $C_L$. 

As shown by Hof in \cite[Example 2.1]{hof1995on}, assuming that the limit 
\begin{equation} \label{freqLimit}
n_a := \lim_{L \to \infty} \frac{N_L(a)}{L^{n}} 
\end{equation} 
exists and is positive as well as finite, for all $a \in X - X$, then the complex measure $\gamma_X$ given by
\begin{equation} \label{auto}
\gamma_{X}(f) := \sum_{a \in A} n_af(a),\quad \text{for all } f \in C_c(\mathbb{R}^n),
\end{equation}
is the unique {\it autocorrelation measure} of $\mu_{X}$. 

If the Fourier transform $\widehat{\gamma}_{X}$ exists as both a measure and a tempered distribution, i.e., as a tempered measure, it is called the {\it diffraction measure} of $X$. We then say that the diffraction measure of $X$ exists. In our setting, according to \cite[Lemma 2.3]{lagarias2000mathematical}, all of our autocorrelation measures will have Fourier transforms that exist as not only tempered distributions, but also as measures. 

One says that a structure $X$ has a {\it pure-point diffraction pattern} if its diffraction measure consists of a purely discrete measure. By leveraging the Poisson Summation Formula, one can show, for example, that a full-rank lattice $\Lambda \subset \mathbb{R}^n$ has a pure point diffraction pattern given by the diffraction measure
\[ \mu = \frac{1}{\det(\Lambda)}\sum_{y \in \Lambda^{\ast}} \delta_y, \]
where $\Lambda^\ast$ is the {\it dual lattice} of $\Lambda$ defined by
\[ \Lambda^\ast = \{x \in \mathbb{R}^n\colon x \cdot y \in \mathbb{Z},\text{ for all } y \in \Lambda\}. \] 
As a special case of this, one sees that the diffraction pattern of the set of integers in $\mathbb{R}$ is again the set of integers. Lagarias has shown in \cite{lagarias2000mathematical} that the diffraction pattern of any full-rank ideal crystal is discrete. This paper extends his result to non full-rank ideal crystals, which have diffraction patterns that are discrete in some dimensions, as we explain in Section \ref{results} below. 

\subsection{Quasiperiodic patterns from complex dimensions} \label{Quasiperiodic patterns from complex dimensions}
\hfill

\noindent Our goal is to address \cite[Problem 3.22]{lapidus2012fractal} (see Problem \ref{open problem} below) by uncovering diffraction properties for the sets of complex dimensions of nonlattice self-similar fractal strings. In the present section, we give some background on self-similar fractal strings, complex dimensions, and the lattice/nonlattice dichotomy in the set of all self-similar fractal strings. In the lattice case, the complex dimensions consist of $1$-dimensional ideal crystals in the plane (see Proposition \ref{idealLattice}); that is, the complex dimensions are periodically distributed (with the same period) along finitely many vertical lines. This is the main motivation behind extending \cite[Theorem 2.7]{lagarias2000mathematical} to the non full rank case. We now have a formula for diffraction by the complex dimensions of arbitrary lattice strings. By contrast, the complex dimensions in the nonlattice case are not periodic as they are in the lattice case. In fact, as is shown in \cite[Section 3.4]{lapidus2012fractal} (see also \cite{lapidus2003complex}), they are approximated by the complex dimensions of lattice strings, in such a way that they are quasiperiodic in a very precise sense. This was most recently studied in \cite{lapidus2021quasiperiodic} by the authors of the present paper, from a numerical perspective.

\begin{problem}[{\cite[Problem 3.22, p. ~89]{lapidus2012fractal}}] \label{open problem}
Is there a natural way in which the quasiperiodic pattern of the set of complex dimensions of a nonlattice self-similar fractal string can be understood in terms of a suitable (generalized) quasicrystal or of an associated quasiperiodic tiling? 
\end{problem}

From 1991 to 1993, Lapidus (in the more general and higher-dimensional case of fractal drums), as well as Lapidus and Pomerance established connections between complex dimensions and the theory of the Riemann zeta function by studying the connection between fractal strings and their spectra; see \cite{lapidus1991fractal}, \cite{lapidus1993vibrations} and \cite{lapidus1993the}. Then, in \cite{lapidus1995the}, Lapidus and Maier used the intuition coming from the notion of complex dimensions in order to rigorously reformulate the Riemann hypothesis as an inverse spectral problem for fractal strings. The notion of complex dimensions was precisely defined and the corresponding rigorous theory of complex dimensions was fully developed by Lapidus and van Frankenhuijsen, for example in \cite{lapidus2000fractal,lapidus2003complex,lapidus2006fractal,lapidus2012fractal}, in the one-dimensional case of fractal strings. Recently, the higher-dimensional theory of complex dimensions was fully developed by Lapidus, Radunovi\'c and {\v Z}ubrini\'c in the book \cite{lapidus2017fractal} and a series of accompanying papers; see also the first author's recent survey article \cite{lapidus2019an}.

\subsubsection{The geometric zeta function and complex dimensions of a fractal string} \label{geometric zeta function section}
An \textit{\textup(ordinary\textup) fractal string} 
$\mathcal{L}$ consists of a bounded open subset $\Omega \subset \mathbb{R}$ (see \cite{lapidus1993the}, \cite{lapidus1993vibrations}, \cite{lapidus2019an}, and, e.g., \cite[Chapter 1]{lapidus2012fractal}). Such a set $\Omega$ is a disjoint union of countably many disjoint open intervals. The lengths, 
\[ \ell_1, \ell_2, \ell_3, \dots, \] 
of the open intervals are called the \textit{lengths} of $\mathcal{L}$, and since $\Omega$ is a bounded set, it is assumed without loss of generality that
\[ \ell_1 \geq \ell_2 \geq \cdots > 0, \]
and that $\ell_j \to 0$ as $j \to \infty$.\footnote{We ignore here the trivial case when $\Omega$ is a finite union of disjoint open intervals.} Let $\sigma_{\mathcal{L}}$ denote the \textit{abscissa of convergence},\footnote{Note that $\lvert \ell_j^s \rvert = \ell_j^{\mathop{\textup{Re}}s}$, for every $s \in \mathbb{C}$ and all $j \in \mathbb{N}$.}  
\begin{equation} \label{abscissa} 
\sigma_{\mathcal{L}} = \inf\left\{\alpha \in \mathbb{R} : \sum_{j = 1}^\infty \ell_j^\alpha < \infty\right\},
\end{equation}
of the \textit{geometric zeta function} 
\[ \zeta_{\mathcal{L}}(s) = \sum_{j = 1}^\infty \ell_j^s \]
of $\mathcal{L}$. Here, the Dirichlet series $\sum_{j \geq 1} \ell_j^s$ converges for $\mathop{\textup{Re}}s > \sigma_{\mathcal{L}}$. Since there are infinitely many lengths, $\zeta_{\mathcal{L}}(s)$ diverges at $s = 0$. Also, since $\Omega$ has finite Lebesgue measure, $\zeta_{\mathcal{L}}(s)$ converges at $s = 1$. Hence, it follows from standard results about general Dirichlet series (see, e.g., \cite{serre1973a}) that the second equality in the above definition of $\sigma_{\mathcal{L}}$ holds, and, therefore, that 
$0 \leq \sigma_{\mathcal{L}} \leq 1$.

\begin{definition} \label{dimension of fractal string}
The {\it dimension} of a fractal string $\mathcal{L}$ associated with a bounded open set 
$\Omega \subset \mathbb{R}$, denoted by $D_\mathcal{L}$, is defined as the \textit{\textup(inner\textup) Minkowski dimension} of $\Omega$:\footnote{Strictly speaking, $D_\mathcal{L}$ is the {\it upper} (inner) Minkowski dimension of $\mathcal{L}$; for the simplicity of exposition, however, we will ignore this nuance here. In the case of a self-similar string, the {\it Minkowski dimension $D_{\mathcal{L}}$ of $\mathcal{L}$ exists}, in the sense that its upper and lower (inner) Minkowski dimensions coincide.}
\[ D_{\mathcal{L}} = \inf\{\alpha \geq 0 \colon V(\varepsilon) = O(\varepsilon^{1 - \alpha})\text{, as }\varepsilon \to 0^+\}, \]
where $V(\varepsilon)$ denotes the volume (i.e., total length, here) of the \textit{inner tubular neighborhood} of $\partial\Omega$ with radius $\varepsilon$ given by
\[ V(\varepsilon) = \operatorname{vol}_1\left(\{x \in \Omega \colon d(x,\partial\Omega) < \varepsilon\}\right). \]   
\end{definition}

According to \cite[Theorem 1.10]{lapidus2012fractal} (see also \cite{lapidus1993vibrations}), the abscissa of convergence 
$\sigma_{\mathcal{L}}$ 
of a fractal string 
$\mathcal{L}$ 
coincides with the dimension 
$D_{\mathcal{L}}$ of $\mathcal{L}$: $\sigma_{\mathcal{L}} = D_{\mathcal{L}}$.

The geometric importance of the set of complex dimensions of a fractal string $\mathcal{L}$ with boundary $\Omega$ (see Definition \ref{complex dimensions} below), which always includes its inner Minkowski dimension $D_{\mathcal{L}}$ (see, e.g., \cite[Chapter 1]{lapidus2012fractal}),\footnote{More specifically, this is so provided $\zeta_{\mathcal{L}}$ admits a meromorphic continuation to a neighborhood of $D_{\mathcal{L}}$. Again, this is the case for all self-similar strings, for example.} can be justified as follows: the complex dimensions appear in an essential way in the explicit formula for the volume $V(\varepsilon)$ of the inner tubular neighborhood of the boundary $\partial\Omega$; see the corresponding ``fractal tube formulas'' obtained in Chapter 8 of \cite{lapidus2012fractal}. 

\begin{definition} \label{complex dimensions}
Suppose $\zeta_{\mathcal{L}}$ has a meromorphic continuation to the entire complex plane. Then the poles of $\zeta_{\mathcal{L}}$ are called the \textit{complex dimensions} of $\mathcal{L}$. 
\end{definition}

Accordingly, the complex dimensions give very detailed information about the intrinsic oscillations that are inherent to fractal geometries; see also Remark \ref{definition of fractal} below. The current paper, however, deals with the complex dimensions viewed only as a discrete subset of the complex plane, and the focus is on diffraction by the sets of complex dimensions of self-similar fractal strings, which are the special type of fractal strings that are constructed through an iterative process involving scaling, as we next explain just after Remark \ref{definition of fractal} below.

\begin{remark} \label{definition of fractal}
A geometric object is said to be {\it fractal} if it has at least one {\it nonreal} complex dimension.\footnote{It then has at least two nonreal complex dimensions since, clearly, nonreal complex dimensions come in complex conjugate pairs.} This definition applies to fractal strings (including all self-similar strings, which are shown to be fractal in this sense \cite{lapidus2000fractal,lapidus2003complex,lapidus2006fractal,lapidus2012fractal}),\footnote{In fact, self-similar strings have infinitely many nonreal complex dimensions; see, e.g., Equation (2.37) in \cite[Theorem 2.16]{lapidus2012fractal}.} and to bounded subsets of 
$\mathbb{R}^n$ (for any integer $n \geq 1$), as well as, more generally, to relative fractal drums, which are natural higher-dimensional counterparts of fractal strings (see \cite{lapidus2017fractal,lapidus2019an}).   
\end{remark}

\subsubsection{Self-similar fractal strings} \label{complex}
\noindent Let $I$ be a nonempty, compact interval with length $L$. Let
\begin{equation} \label{contractions}
\Phi_1, \Phi_2, \dots, \Phi_M:I \to I
\end{equation}
be $M \geq 2$ contraction similitudes with distinct scaling ratios
\[ 1 > r_1 > r_2 > \cdots > r_N > 0\qquad (M \geq N).\footnote{Accordingly, the corresponding, and not necessarily distinct, scaling ratios would be repeated according to their respective multiplicities.} \]
What this means is that for all $j = 1, \dots, M$,
\[ \left\rvert \Phi_j(x) - \Phi_j(y) \right\lvert = r_j\lvert x - y \rvert, \text{ for all } x, y \in I. \]
Assume that after having applied to $I$ each of the maps $\Phi_j$ in (\ref{contractions}), for $j = 1, \dots, M$, the resulting images,
\begin{equation} \label{images}
\Phi_1(I), \dots, \Phi_M(I),
\end{equation}
do not overlap, except possibly at the endpoints, and that $\sum_{j = 1}^M r_j < 1$.

These assumptions imply that the complement of the union, $\bigcup_{j = 1}^M \Phi_j(I)$, in $I$ consists of $K \geq1$ pairwise disjoint open intervals with lengths 
\[ 1 > g_1L \geq g_2L \geq \cdots \geq g_KL > 0, \]
called the {\it first intervals}. Note that the quantities $g_1, \dots, g_K$, called the \textit{gaps}, along with the scaling ratios $r_1, \dots, r_M$, satisfy the equation 
\[ \sum_{j = 1}^M r_j + \sum_{k = 1}^K g_k = 1. \]

The process which was just described above is then repeated for each of the $M$ images $\Phi_j(I)$, in (\ref{images}), in order to produce $KM$ additional pairwise disjoint open intervals in $I$. Repeating this process ad infinitum yields countably many open intervals, which defines a fractal string $\mathcal{L}$ with bounded open set $\Omega$ given by the (necessarily disjoint) union of these open intervals. Any fractal string obtained in this manner is called a \textit{self-similar fractal string} (or a {\it self-similar string}, in short).
 
The first two authors of the present paper have shown that the geometric zeta function of any self-similar fractal string has a meromorphic continuation to all of $\mathbb{C}$; see Theorem 2.3 in Chapter 2 of \cite{lapidus2012fractal}. Specifically, the geometric zeta function $\zeta_\mathcal{L} = \zeta_\mathcal{L}(s)$ of any self-similar fractal string with scaling ratios $\{r_j\}_{j = 1}^M$, gaps $\{g_k\}_{k = 1}^K$, and total length $L$ is given by

\begin{equation} \label{extension} 
\zeta_\mathcal{L}(s) = \frac{L^s\sum_{k = 1}^{K} g_k^s}{1 - \sum_{j = 1}^{M} r_j^s}, \text{ for all } s \in \mathbb{C}.
\end{equation}

Note that both the numerator and the denominator of the right-hand side of (\ref{extension}) are special kinds of exponential polynomials, known as Dirichlet polynomials. 

\begin{definition} \label{Dirichlet polynomial}
Given an integer $N \geq 1$, let $r_0 > r_1 > \cdots > r_N > 0$, and let $m_0, m_1, \dots, m_N \in \mathbb{C}$. The function $f: \mathbb{C} \to \mathbb{C}$ given by 
\begin{equation} \label{poly} 
f(s) = \sum_{j = 0}^N m_jr_j^s
\end{equation}
is called a \textit{Dirichlet polynomial} with {\it scaling ratios} 
$r_1, \dots, r_N$
and respective {\it multiplicities} 
$m_0, \dots, m_N$.\footnote{In the geometric situation of a self-similar string $\mathcal{L}$ discussed just above, $r_0 := 1$ and $m_0 := -1$, while the $r_j$'s, with $j = 1, \dots, N$, correspond to the {\it distinct} scaling ratios, among the scaling ratios $\{r_j\}_{j = 1}^M$ of $\mathcal{L}$. Hence, in particular, $1 \leq N \leq M$ in this case, and, modulo a suitable abuse of notation, for each distinct scaling ratio $r_j$, for $j = 1, \dots, M$, $m_j := \#\{1 \leq k \leq M\colon r_k = r_j\}$ is indeed the multiplicity of $r_j$.} 
\end{definition}

Therefore, the set of complex dimensions of any self-similar fractal string is a subset of the set of complex roots of an associated Dirichlet polynomial $f = f(s)$ with positive integer multiplicities. While, in general, some of the zeros of the denominator $f = f(s)$ of the right-hand side of (\ref{extension}) could be cancelled by the roots of its numerator (see \cite[Section 2.3.3]{lapidus2012fractal}), in the important special case of a single gap length 
(i.e., when $g_1 = \cdots = g_K$), the complex dimensions precisely coincide with the complex roots of $f$. 

Let $\mathcal{L}$ be a self-similar fractal string with the associated Dirichlet polynomial  
\[ f(s) = \sum_{j = 0}^N m_jr_j^s,\text{ for all }s \in \mathbb{C}. \]
Assume without loss of generality that $r_0 = 1$. Define the {\it weights} $w_1, \dots, w_N$ of $f$ by $w_j: = -\log r_j$, for $1 \leq j \leq N$. From the perspective of the current paper, there is an important dichotomy in the set of all self-similar fractal strings, according to which any self-similar fractal string is either {\it lattice} or {\it nonlattice}, depending on the scaling ratios with which a self-similar fractal string is constructed. Specifically, the {\it lattice case\/} is when the $N$ distinct weights $w_j$ are rationally dependent, and the {\it nonlattice case\/} is when at least one ratio $w_i/w_j$ is irrational. In the lattice case, the multiplicative group $G$ generated by the scaling ratios $r_j$ is of rank 1 (that is, $G = \langle r_1, \dots, r_N \rangle = r^{\mathbb{Z}}$, for some $r \in (0,1)$, called the multiplicative generator), and this group is of rank at least $2$ in the nonlattice case. The {\it generic nonlattice case\/} is when $N \geq 2$ and the rank of $G$ is equal to $N$.

\begin{definition} \label{lattice/nonlattice}
A Dirichlet polynomial $f$ is called {\it lattice} if $w_j / w_1$ is rational for $1 \leq j \leq N$, and it is called {\it nonlattice} otherwise.\footnote{Note that if $N = 1$, then $f$ must be lattice because $w_1/w_1 = 1$ is rational.}
\end{definition}

In the lattice case, there exist positive integers
\[ k_1 < k_2 < \cdots < k_N < \infty \]
such that 
\[ \frac{w_j}{w_1} = \frac{k_j}{k_1}, \quad \text{for } j = 2, \dots, N. \]

According to \cite[Theorem 3.6]{lapidus2012fractal}, the complex roots of a lattice Dirichlet polynomial 
$f$ lie periodically on finitely many vertical lines, and on each line they are separated by the positive number 
\[ \textbf{p} = \frac{2\pi}{\log r^{-1}}, \] 
called the \textit{oscillatory period of $f$}. 

More precisely, following the discussion surrounding \cite[Equation (2.48), p. 58]{lapidus2012fractal}, the roots are computed by first rewriting $f(s)$ as a polynomial $g(z)$ of degree $k_N$ in the complex variable $z := r^s$, where $r := r_1^{1/k_1}$ is the multiplicative generator of $f = f(s)$: 
\begin{equation} \label{complex poly}
g(z) = 1 - m_1z^{k_1} - m_2z^{k_2} - m_3z^{k_3} - \cdots - m_Nz^{k_{N}}.
\end{equation}
There are $k_N$ roots of $g(z)$, counted with multiplicity. Each root is of the form 
\[ z = \lvert z \rvert e^{i\theta}, \]
where $-\pi < \theta \leq \pi$, and it corresponds to a unique root of $f(s)$, namely,
\[ s = \omega = \frac{-\log \lvert z \rvert}{\log r^{-1}} - \frac{i \theta}{\log r^{-1}}. \]
Therefore, given a lattice Dirichlet polynomial $f$ with oscillatory period $\textbf{p}$, there exist finitely many {\it principal complex roots} 
\[ W(k_N) := \{\omega_1, \dots, \omega_{k_N}\}, \] 
such that the set $\mathcal{D}_f$ of complex roots of $f = f(s)$ is given by
\[ \mathcal{D}_f = \bigcup_{1 \leq j \leq k_N} H_j, \]
where for $1 \leq j \leq k_N$,
\[ H_j := \{\omega_j + in\textbf{p} : n \in \mathbb{Z}\}. \]

The set of complex roots of any Dirichlet polynomial is a subset of the horizontally bounded vertical strip 
\[ R := \{z \in \mathbb{C} : D_\ell \leq \mathop{\textup{Re}} z \leq D\}, \] 
where 
$D_\ell$ 
and 
$D$ 
are the unique real numbers satisfying the equations\footnote{If $N = 1$, then the first sum on the left-hand side of (\ref{sums}) is equal to zero, by convention.}
\begin{equation} \label{sums}
1 + \sum_{j = 1}^{N - 1} |m_j|r_j^{D_\ell} = |m_{N}|r_N^{D_\ell}\quad \text{and}\quad \sum_{j = 1}^{N} |m_j|r_j^D = 1,
\end{equation}
respectively. These numbers satisfy the inequality $-\infty < D_{\ell} \leq D$. 
Since the multiplicities are assumed to be positive integers, the complex roots are symmetric about the real axis, the number $D$ defined above is positive, and it is the only real root of $f = f(s)$; furthermore, it is a simple root.

\begin{remark}
In the case of a self-similar string $\mathcal{L}$, the nonnegative number $\omega_{k_N} = D$, $D$ does not exceed 1, and $D$ coincides with $D_{\mathcal{L}}$, the inner Minkowski dimension of $\mathcal{L}$: $D = D_{\mathcal{L}} = \sigma_{\mathcal{L}}$, in the notation introduced in (\ref{abscissa}) and Definition \ref{dimension of fractal string}.
\end{remark}

The complex dimensions of a lattice polynomial\footnote{See Definition \ref{complex dimensions} above for the definition of the complex dimensions as the poles of the `geometric zeta function' associated with a fractal string.} can be numerically obtained via the roots of certain polynomials that are typically sparse with large degrees, and lie periodically on finitely many vertical lines counted according to multiplicity. Furthermore, on each vertical line, they are separated by a positive real number ${\bf p}$, called the {\it oscillatory period} of the string. (See \cite[Chapter 2]{lapidus2000fractal}, \cite[Theorem 2.5]{lapidus2003complex}, and \cite[Theorems 2.16 and 3.6]{lapidus2012fractal}.)

\begin{example} \label{complexEx}
Consider the self-similar fractal string $\mathcal{L}$ constructed from the contraction similitudes
\[ \Phi_1(x) = \frac{1}{2}x,\quad \Phi_2(x) = \frac{1}{8}x + \frac{7}{3}, \]
and the initial interval $I = [0,8/3]$ with length $L = 8/3$. We have two distinct scaling ratios given by $r_1 = 1/2$ and $r_2 = (1/2)^3$, and one gap given by $g = 1$. In order to obtain the complex dimensions, we compute the poles of the geometric zeta function given in (\ref{extension}),
\begin{equation}
\zeta_\mathcal{L}(s) = \frac{\left(8/3\right)^s}{1 - 2^{-s} - 8^{-s}}.
\end{equation}
Thus, the set of complex dimensions of this lattice string consists of the set of roots of the Dirichlet polynomial 
\[ f(s) = 1 - 2^{-s} - 8^{-s}. \]
which are computed by solving the polynomial equation
\[ 1 - z - z^3 = 0, \]
where $z := 2^{-s}$. A plot of the roots, which is also the set of complex dimensions of the lattice string, is shown in Figure \ref{rootPlot}.

\begin{figure}[tb] 
\centering
\includegraphics[width=0.75\textwidth]{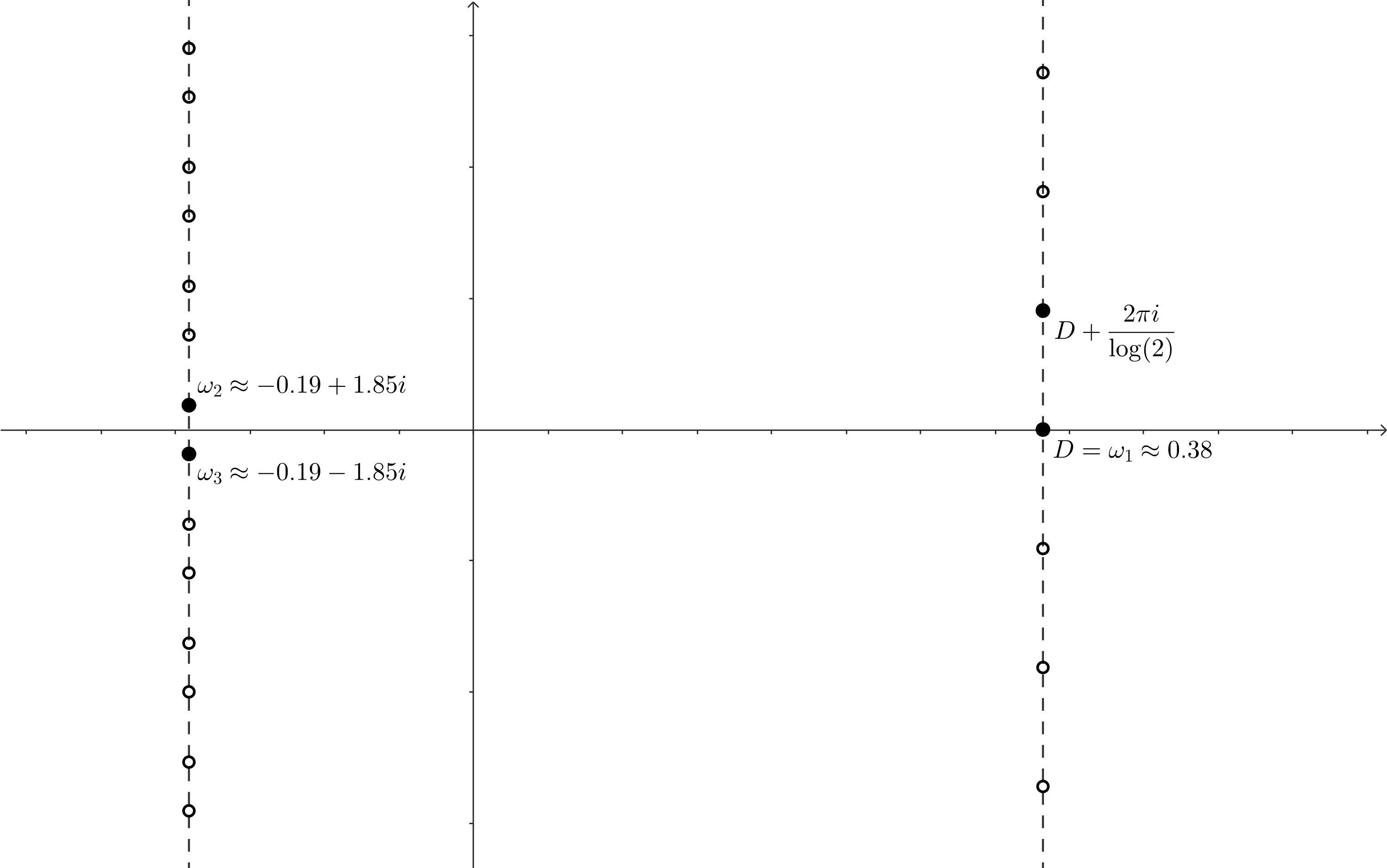}
\caption{The complex dimensions of the lattice string from Example \ref{complexEx}.} 
\label{rootPlot}
\end{figure}
\end{example} 

For nonlattice self-similar fractal strings, the complex dimensions cannot be numerically obtained in the same way as in the lattice case. Indeed, they correspond to the roots of a transcendental (rather than polynomial) equation. They can, however, be approximated by the complex dimensions of a sequence of lattice strings with larger and larger oscillatory periods.  The {\it Lattice String Approximation algorithm} of Lapidus and van Frankenhuijsen, referred to in this paper as the LSA {\it algorithm} (and theorem), allows one to replace the study of nonlattice self-similar fractal strings by the study of suitable approximating sequences of lattice self-similar fractal strings. Using this algorithm, Lapidus and van Frankenhuijsen have shown that the sets of complex dimensions of nonlattice self-similar fractal strings are quasiperiodically distributed, in a precise sense (see, e.g., \cite[Theorem 3.6, Remark 3.7]{lapidus2003complex} and \cite[Subsection 3.4.2]{lapidus2006fractal,lapidus2012fractal}), and they have illustrated their results by means of a number of examples (see, e.g., the examples from Section 7 in \cite{lapidus2003complex} and their counterparts in Chapters 2 and 3 of \cite{lapidus2012fractal}). 

Following the suggestion by those same authors in the introduction of \cite{lapidus2003complex}, and in \cite[Remark 3.38]{lapidus2006fractal,lapidus2012fractal}, the authors of the present paper presented in \cite{lapidus2021quasiperiodic} an efficient implementation of the LSA algorithm incorporating the application of a powerful lattice basis reduction algorithm, which is due to A. K. Lenstra, H. W. Lenstra and L. Lov\'asz and is known as the LLL {\it algorithm}, in order to generate simultaneous Diophantine approximations; see \cite[Proposition 1.39]{lenstra1982factoring} and \cite[Proposition 9.4]{bremner2011lattice}. It also uses the open source software MPSolve, due to D. A. Bini, G. Fiorentino and L. Robol in \cite{bini2000design}, \cite{bini2014solving}, in order to approximate the roots of large degree sparse polynomials. Indeed, the LLL algorithm along with MPSolve allow for a deeper numerical and visual exploration of the quasiperiodic patterns of the complex dimensions of self-similar strings via the LSA algorithm than what had already been done in \cite{lapidus2000fractal, lapidus2003complex, lapidus2006fractal, lapidus2012fractal}.

In the latter part of \cite[Chapter 3]{lapidus2012fractal}, a number of mathematical results were obtained concerning either the nonlattice case with two distinct scaling ratios (amenable to the use of continued fractions) as well as the nonlattice case with three or more distinct scaling ratios (therefore, typically requiring more complicated simultaneous Diophantine approximation algorithms). In \cite{lapidus2021quasiperiodic}, it has become possible, in particular, to explore more deeply and accurately additional nonlattice strings with rank greater than or equal to three, i.e., those that cannot be solved using continued fractions. We will address the nonlattice case in the sequel to this paper, \cite{lapidus2024diffraction}.

\section{The Diffraction Measure of a Lattice String} \label{results}

Lagarias computed the diffraction measure for a full-rank ideal crystal, showing that an ideal crystal has a pure-point discrete diffraction measure \cite[Theorem 2.7]{lagarias2000mathematical}. Proposition \ref{generalPSF} below gives our derivation of a suitable version of the Poisson Summation Formula (PSF) that is amenable to an $m$-dimensional lattice $\Lambda \subset \mathbb{R}^{d + m}$, where $d$ is a nonnegative integer. We note that while the version of the PSF we use in the present paper is already known (see, e.g., \cite[Proposition 6.2, Corollary 7.1]{argabright1974fourier} and \cite[Theorem 5.5.2]{reiter2000classical}), we have presented a proof that closely resembles that of the classical PSF (see, e.g., \cite[Chapter 9]{baake2013aperiodic}) for the benefit of the reader. 
Using Proposition \ref{generalPSF}, and by adjusting the averaging shapes in the computation of the autocorrelation measure, we obtain a generalization of \cite[Theorem 2.7]{lagarias2000mathematical} to the case of any (generalized) $m$-dimensional ideal crystal in $\mathbb{R}^n$.

Let $\mathcal{S}(\mathbb{R}^n)$ denote the space of Schwartz functions on $\mathbb{R}^n$, consisting of all infinitely differentiable functions on $\mathbb{R}^n$ having rapid decay at infinity; see, e.g., \cite{schwartz1998theorie}. Note that for a (suitable) function 
\[ F\colon \mathbb{R}^d \oplus \mathbb{T}^m \to \mathbb{C}, \]
we have the Fourier inversion formula
\[ F(x,\alpha) = \sum_{b \in \mathbb{Z}^m} \int_{\mathbb{R}^d} \widehat{F}(y,b)e(x \cdot y + \alpha \cdot b)\, dy, \]
where for $y \in \mathbb{R}^d$ and $b \in \mathbb{Z}^m$,
\[ \widehat{F}(y,b) = \int_{\mathbb{T}^m} \int_{\mathbb{R}^d} F(x,\alpha)e(-y \cdot x - b \cdot \alpha)\, dx\, d\alpha, \]
and $e(t)$ denotes $e^{2\pi it}$ (see, e.g., \cite[Section VII,4]{katznelson2004an}, \cite{lapidus2008in}, or \cite{serre1973a}). Here, and thereafter, $\mathbb{T}^m := \mathbb{R}^m/\mathbb{Z}^m$ stands for the $m$-dimensional torus --- viewed, for us, as a compact abelian group. 

\begin{proposition} \label{generalPSF}
For any $\phi \in \mathcal{S}(\mathbb{R}^{d} \oplus \mathbb{R}^m)$ and point $(x,\alpha) \in \mathbb{R}^d \oplus \mathbb{T}^m$, we have
\begin{equation} \sum_{a \in \mathbb{Z}^m} \phi(x,a + \alpha) = \sum_{b \in \mathbb{Z}^m} \int_{\mathbb{R}^d} \widehat{\phi}(y,b)e(x \cdot y + \alpha \cdot b)\, dy.
\end{equation}
\end{proposition}
\begin{proof}
Let $\phi \in \mathcal{S}(\mathbb{R}^d \oplus \mathbb{R}^m)$ and define the function $F\colon \mathbb{R}^d \oplus \mathbb{T}^m \to \mathbb{C}$ by
\[ F(x,\alpha) = \sum_{a \in \mathbb{Z}^m} \phi(x,a + \alpha). \]
We first compute $\widehat{F}(y,b)$ for $b \in \mathbb{Z}^m$, and show that it is equal to $\widehat{\phi}(y,b)$:
\begin{align*} 
\widehat{F}(y,b) &= \int_{\mathbb{T}^m} \int_{\mathbb{R}^d} F(\alpha,x)e(-y \cdot x - b \cdot \alpha)\, dx\, d\alpha \\
&= \int_{\mathbb{T}^m} \int_{\mathbb{R}^d} \sum_{a \in \mathbb{Z}^m} \phi(x,a + \alpha)e(-y \cdot x - b \cdot \alpha)e(-b \cdot a)\, dx\, d\alpha,
\end{align*}
as $e(-b \cdot a) = 1$. Next, the sum over $a$ and the integral over the torus gives an integral over $\mathbb{R}^m$,
\[ \widehat{F}(y,b) = \int_{\mathbb{R}^d} \sum_{a \in \mathbb{Z}^m} \int_{\mathbb{T}^m} \phi(x,a + \alpha)e(-y \cdot x - b \cdot (a + \alpha))\, d\alpha\, dx = \widehat{\phi}(y,b). \]
Now, by Fourier inversion, we obtain the identity
\[ F(x,\alpha) = \sum_{b \in \mathbb{Z}^m} \int_{\mathbb{R}^d} \widehat{\phi}(y,b) e(x \cdot y + \alpha \cdot b)\, dy. \]
This completes the proof of the proposition.
\end{proof}

\begin{theorem}[Extended Poisson Summation Formula: General Version] \label{LatticegeneralPSF}
For a basis $B$ of $\mathbb{R}^m$, let $\Lambda = B\mathbb{Z}^m$ be a full-rank lattice in $\mathbb{R}^m$. Let $d \geq 0$ be an integer. Then, for any $\phi \in \mathcal{S}(\mathbb{R}^{d} \oplus \mathbb{R}^m)$, and every $\alpha \in \mathbb{R}^m / \Lambda$ and $x \in \mathbb{R}^d$,
\begin{equation} \label{EPF}
\sum_{l \in \Lambda} \phi(x,l + \alpha) = \frac{1}{\left\lvert\det(B)\right\rvert}\sum_{b \in \left(B^{-1}\right)^T\mathbb{Z}^m}\int_{\mathbb{R}^d} \widehat{\phi}(y,b)e(x \cdot y + \alpha \cdot b)\, dy,
\end{equation} 
where $\left(B^{-1}\right)^T\mathbb{Z}^m$ is the lattice dual to $\Lambda$, generated by the transpose of $B^{-1}$.
\end{theorem}
\begin{proof}
Let $\phi \in \mathcal{S}(\mathbb{R}^{d} \oplus \mathbb{R}^m)$, and write
\[ \sum_{l \in \Lambda} \phi(x,l + \alpha) = \sum_{a \in \mathbb{Z}^m} \left(\phi \circ \left(\text{I}_{\mathbb{R}^d} \oplus B\right)\right)(x,a + B^{-1}\alpha), \]
where $\text{I}_{\mathbb{R}^d}$ denotes the identity map of $\mathbb{R}^d$. The Fourier transform of the function in the sum above is computed by the substitution $u := Bt$,
\begin{align*}
\left(\phi \circ \left(\text{I}_{\mathbb{R}^d} \oplus B\right)\right)^{\widehat{}}(y,b) &= \int_{\mathbb{R}^m} \int_{\mathbb{R}^d} \phi(x,Bt)e(- y \cdot x - b \cdot t)\, dx\, dt \\
&= \frac{1}{\lvert \det(B) \rvert}\int_{\mathbb{R}^m} \int_{\mathbb{R}^d} \phi(x,u)e(- y \cdot x - b \cdot B^{-1}u)\, dx\, du.
\end{align*}
Next, we note that $b \cdot B^{-1}u = b^TB^{-1}u$ is a matrix product; hence $b \cdot B^{-1}u = (B^{-1})^Tb \cdot u $. We then obtain  
\begin{align*}
\left(\phi \circ \left(\text{I}_{\mathbb{R}^d} \oplus B\right)\right)^{\widehat{}}(y,b) &= \frac{1}{\lvert \det(B) \rvert}\int_{\mathbb{R}^m} \int_{\mathbb{R}^d} \phi(x,u)e(- y \cdot x - (B^{-1})^Tb \cdot u)\, dx\, du \\
&= \frac{1}{\lvert \det(B) \rvert}\widehat{\phi}(y,(B^{-1})^Tb).
\end{align*}
Proposition \ref{generalPSF} then gives the desired formula.
\end{proof}

\begin{remark} \label{remark full}
Note that the classical version of the Poisson Summation Formula corresponds to a {\it nondegenerate lattice}, i.e., when $d = 0$ and in other words to the case when the lattice $\Lambda$ is of full rank in all of $\mathbb{R}^d \oplus \mathbb{R}^m$.
\end{remark}

Inspired by \cite[Remark 9.2]{baake2013aperiodic}, we give a suitable version of the classical definition of the autocorrelation measure reviewed in Subsection \ref{auto section}, so that it can be applied to an ideal crystal generated from a rank-$m$ lattice in $(d + m)$-dimensional Euclidean space, in the case that $d > 0$. More specifically, we replace the centered cubes in (\ref{count}) with rectangles of fixed width that grow only vertically to cover the lattice. This is where our definition of autocorrelation deviates from Hof's in \cite{hof1992quasicrystals}-\cite{hof1995diffraction}.   

As an example, consider the planar ideal crystal from Figure \ref{ideal crystal} given by
\[ X = \{0\} \times \mathbb{Z} + F, \]
where $F = \{(3/2,-5/4),(0,-1/2),(1,0)\}$.
We do not consider Hof's definition of autocorrelation, which averages out by the centered square $C_L$ with side $L$, as this would give us 
\[ n_a = \lim_{L \to \infty} \frac{N_L(a)}{L^2} = 0 \] 
for all $a \in X - X$, yielding an autocorrelation measure that is equal to 0. Instead, we use the centered rectangle $R_{C,L}$ given by $R_{C,L} = [-C,C] \times [-L,L]$, where $C := \max_{f \in F}\|\pi(f)\|_\infty$,\footnote{In case this maximum vanishes, we use $C=1$.} and $\pi: \mathbb{R}^2 \to \mathbb{R}$ is the projection onto the horizontal axis. This is illustrated in Figure \ref{autoEX}. 

\begin{figure}[tb]
\centering
\includegraphics[width=0.5\textwidth]{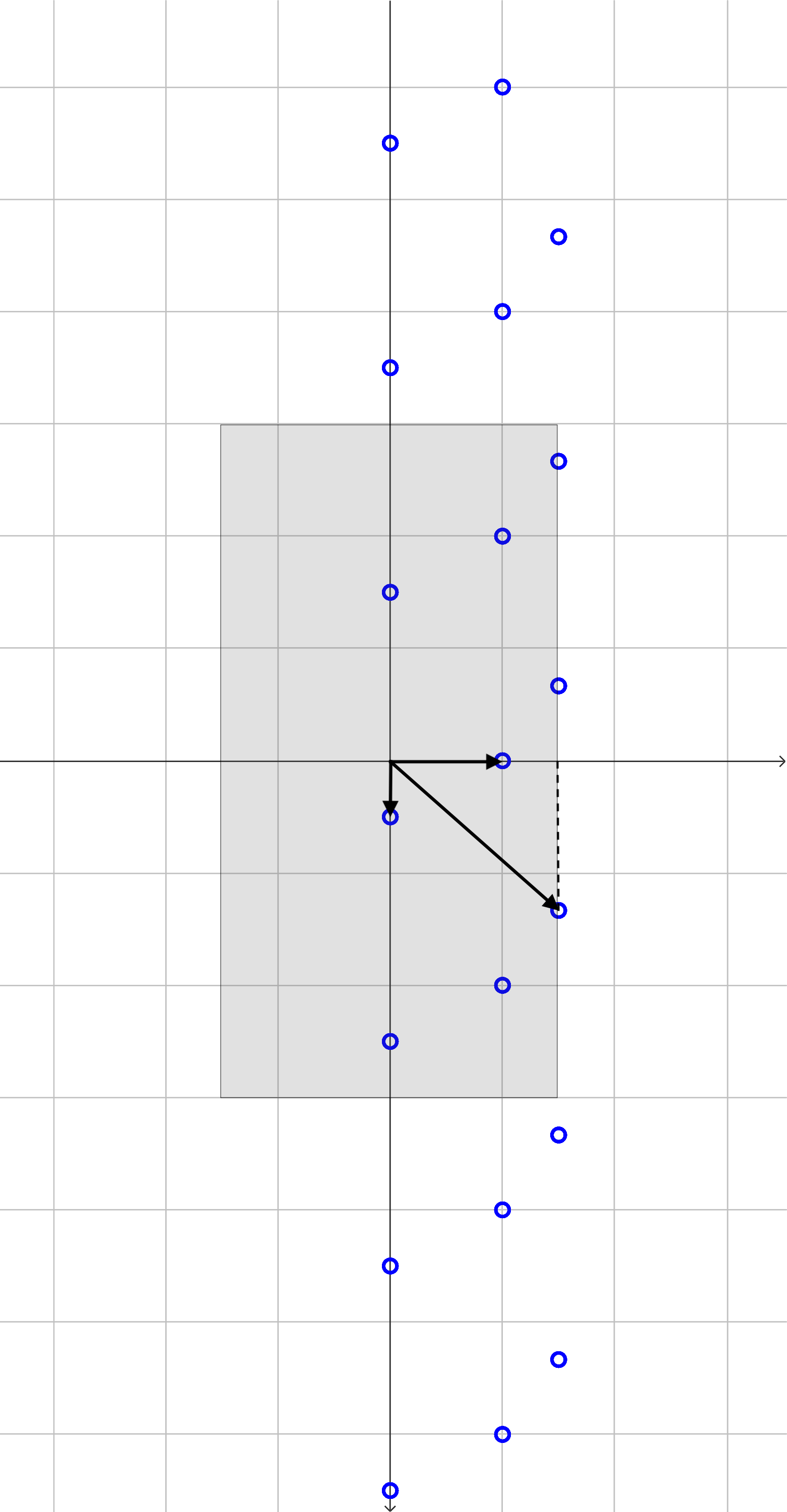}
\caption{The rectangle $R_{3,3/2} = [-3/2,3/2] \times [-3,3]$.}
\label{autoEX}
\end{figure}

More generally, if $\Lambda = B\mathbb{Z}^m$ is a lattice in $\mathbb{R}^m$ with $F \subset \mathbb{R}^{d + m}$ consisting of a finite set of translation vectors, we average out by the centered $(d + m)$-dimensional rectangle $R_{C,L}$ with volume equal to $(2C)^d(2L)^m$, where the first $m$ sides of $R_{C,L}$ are equal to $2L$, and the lengths of the remaining $d$ sides are fixed to being equal to $2C$, where 
\begin{equation} \label{C} 
C := \max_{f \in F}\|\pi_d(f)\|_\infty.
\end{equation}

\begin{definition}
Let $X$ be a countably infinite set in $\mathbb{R}^{d + m}$. If the limit
\begin{equation} \label{nCa} 
n_{C,a} := \lim_{L \to \infty} \frac{N_{C,L}(a)}{(2C)^d(2L)^m}
\end{equation}
exists and is positive as well as finite, for all $a \in X - X$, where 
\[ N_{C,L}(a) := \#\{(x_1,x_2) \in X \times X\colon x_1, x_2 \in R_{C,L} \text{ and } a = x_2 - x_1\}, \]
then the complex measure $\gamma$ given by
\[ \gamma(\phi) := \sum_{a \in X - X} n_{C,a}\phi(a),\quad \text{for all } \phi \in \mathcal{S}(\mathbb{R}^{d + m}), \]
is the unique {\it generalized autocorrelation measure} of $X$.
\end{definition}

\begin{theorem}[Autocorrelation and Diffraction Measures of a Possibly Degenerate Ideal Crystal] \label{idealDiff}
Let $\Lambda = B\mathbb{Z}^m$ be a lattice in $\mathbb{R}^m$. Let $d \geq 0$ be an integer, and let $F \subset \mathbb{R}^{d + m}$ be a finite set of translation vectors. Then, the generalized autocorrelation measure of the $(d + m)$-dimensional ideal crystal 
\[ X = T(\Lambda) + F \]
exists, is unique, and is given, for any $\phi \in \mathcal{S}(\mathbb{R}^{d + m})$, by
\begin{equation} \label{autoc}
\gamma(\phi) = \frac{1}{(2C)^d\lvert \det(B) \rvert}\sum_{a \in \Lambda} \left(\sum_{f,g \in F} \phi(T(a) + (f - g)) \right),
\end{equation}
where $C$ is given by (\ref{C}) and $\pi_d: \mathbb{R}^{d + m} \to \mathbb{R}^d$ is the projection onto the first coordinate.  

Furthermore, the Fourier transform of $\gamma$ is given, for any $\phi \in \mathcal{S}(\mathbb{R}^d \oplus \mathbb{R}^m)$, by
\begin{multline} \label{diff}
\widehat\gamma(\phi) = \frac{1}{(2C)^d\lvert \det(B) \rvert^2} \cdot \\
\int_{\mathbb{R}^{d}}\sum_{b \in \left(B^{-1}\right)^T\mathbb{Z}^m}\left(\sum_{f,g \in F} e(\pi_d(f - g) \cdot y + \pi_m(f - g) \cdot b)\right)\phi(y,b)\, dy.
\end{multline}
We conclude that the diffraction measure $\widehat\gamma$ of a {\it possibly} degenerate ideal crystal exists, is unique, and is given by \textup(\ref{diff}\textup).
\end{theorem}
\begin{proof}
We compute the unique generalized autocorrelation measure (\ref{autoc}) following a similar approach as the one outlined in \cite[Theorem 2.7]{lagarias2000mathematical}. Specifically, as explained in the text just after Remark \ref{remark full}, we use the formula (\ref{auto}) but replace the term $n_a$ with (\ref{nCa}). As noted by Lagarias, the points of $X - X$ coincide with $T(\Lambda) + (F - F)$ and we get
\[ n_{C,a} = \frac{1}{(2C)^d\lvert \det(B) \rvert} \]
for all $a \in X - X$. To prove (\ref{diff}), we express $\widehat\gamma(\phi)$ as follows,
\begin{align*}
\widehat\gamma(\phi) = \gamma(\widehat\phi) &= \frac{1}{(2C)^d\lvert \det(B) \rvert}\sum_{a \in \Lambda} \left(\sum_{f,g \in F} \widehat\phi(T(a) + (f - g)) \right) \\
&= \frac{1}{(2C)^d\lvert \det(B) \rvert}\sum_{a \in \Lambda} \left(\sum_{f,g \in F} \widehat\phi(\pi_d(f - g),a + \pi_m(f - g)) \right).
\end{align*}
Then Theorem \ref{LatticegeneralPSF} yields the result.
\end{proof}

\begin{proposition} \label{idealLattice}
The set of complex dimensions of any lattice self-similar fractal string is a rank-1 generalized ideal crystal in $\mathbb{C} \simeq \mathbb{R}^2$.
\end{proposition}
\begin{proof}
As was recalled in Subsection \ref{complex} above, for any lattice self-similar fractal string $\mathcal{L}$ with oscillatory period {\bf p}, it follows from \cite[Theorem 3.6]{lapidus2012fractal} that there exist finitely many complex numbers, $\omega_1, \omega_2, \dots, \omega_d$, such that the complex dimensions of $\mathcal{L}$ are given by
\[ \mathcal{D}_{\mathcal{L}} = \{\omega_u + in{\bf p}\colon n \in \mathbb{Z}, u = 1, \dots, d \}. \]
Hence, $\mathcal{D}_{\mathcal{L}}$ is a generalized ideal crystal of rank 1.
\end{proof}

In light of Proposition \ref{idealLattice}, the following key result is an immediate corollary of our main theorem (Theorem \ref{idealDiff} just above), in the special case when $m = d= 1$. (At this point, the reader may wish to review some of the notation and results recalled in Section \ref{Quasiperiodic patterns from complex dimensions} above.)  

\begin{corollary}[Autocorrelation and Diffraction Measures of the Complex Dimensions of a Lattice String] \label{complexDiff}
Let $\mathcal{L}$ be a lattice self-similar fractal string with oscillatory period ${\bf p}$, and with complex dimensions generated by
\[ W(k_N) = \{\omega_1, \omega_2, \dots, \omega_{k_N} = D\} \subset \mathbb{C}, \]
where $D = D_{\mathcal{L}}$ is the Minkowski dimension of $\mathcal{L}$.

Then, the autocorrelation measure $\gamma$ of the set of complex dimensions of $\mathcal{L}$ exists, is unique, and is given, for any $\phi \in \mathcal{S}(\mathbb{R}^2)$, by
\begin{equation}
\gamma(\phi) = \frac{1}{2C{\bf p}}\sum_{a \in {\bf p}\mathbb{Z}} \left(\sum_{j = 1}^{k_N}\sum_{k = 1}^{k_N} \phi(\mathop{\textup{Re}(\omega_j - \omega_k),a + \mathop{\textup{Im}}(\omega_j - \omega_k)})\right),
\end{equation}
where $C = \max\{D,D_l\}$.

The diffraction measure $\widehat{\gamma}$ of the set of complex dimensions also exists, is unique, and is given by
\begin{equation} 
\widehat\gamma(\phi) = \frac{1}{2C{\bf p}^2}\sum_{b \in {\bf p}^{-1}\mathbb{Z}}\int_{\mathbb{R}}\left(\sum_{j = 1}^{k_N}\sum_{k = 1}^{k_N} e(\mathop{\textup{Re}}(\omega_j - \omega_k)y + \mathop{\textup{Im}}(\omega_j - \omega_k)b)\right)\phi(y,b)\, dy, 
\end{equation}
still for any $\phi \in \mathcal{S}(\mathbb{R}^2)$.
\end{corollary}

\section{Acknowledgments}
The research of M. L. L. was supported by the Burton Jones Endowed Chair in Pure Mathematics, as well as by the NSF award DMS-1107750.  We thank Alan Haynes for his ideas and suggestions. We also thank the reviewers for their constructive and helpful comments, which have significantly helped us to improve the quality and clarity of our manuscript.

\end{document}